\documentclass[12pt]{article}%
\usepackage{graphicx}
\usepackage{amsmath}
\usepackage{amsfonts}
\usepackage{amssymb}%
\newtheorem{theorem}{Theorem}

\newtheorem{condition}[theorem]{Condition}

\newtheorem{lemma}[theorem]{Lemma}

\newenvironment{proof}[1][Proof]{\noindent\textbf{#1.} }{\ \rule{0.5em}{0.5em}}
\begin{document}

\title{\textbf{Absence of Breakdown of the Poisson Hypothesis }\\\textbf{I. Closed Networks at Low Load }}
\author{Alexander Rybko$^{1}$, Senya Shlosman$^{1,2}$, Alexander Vladimirov$^{1}$\\$^{1}$ Inst. of the Information Transmission Problems,\\Russian Academy of Sciences,\\Moscow, Russia\\$^{2}$ Centre de Physique Theorique, UMR 6207,\\CNRS, Luminy, Marseille, France}
\maketitle

\begin{abstract}
We prove that the general mean-field type networks at low load behave in
accordance with the Poisson Hypothesis. That means that the network
equilibrates in time independent of its size. This is a \textquotedblleft
high-temperature\textquotedblright\ counterpart of our earlier result, where
we have shown that at high load the relaxation time can diverge with the size
of the network (\textquotedblleft low-temperature\textquotedblright). In other
words, the phase transitions in the networks can happen at high load, but
cannot take place at low load.


MSC-class: 82C20 (Primary), 60J25 (Secondary)

\end{abstract}

\section{Introduction}

The Poisson Hypothesis is a device to predict the behavior of large queuing
networks. It was formulated first by L. Kleinrock, and concerns the following situation.

Suppose we have a large network of servers, through which many customers are
traveling, being served at different nodes of the network. If the node is
busy, the customers wait in the queue. Customers are entering into the network
from the outside via some nodes, and these external flows of customers are
Poissonian, with constant rates. The service time at each node is random,
depending on the node, and the customer.

We are interested in the stationary distribution $\pi_{\mathcal{N}}$ at a
given node $\mathcal{N}$: what is the distribution of the queue at
$\mathcal{N}$, what is the average waiting time, etc.

Except for a very few special cases, when the service times are exponential,
the distributions $\pi_{\mathcal{N}}$ in general can not be computed. The
recipe of the Poisson Hypothesis for approximate computation of $\pi
_{\mathcal{N}}$ and the prediction for the (long-time, large-size) behavior of
the network is the following:

\begin{itemize}
\item Consider the total flow $\mathcal{F}$ of customers to a given node
$\mathcal{N}.$ Then $\mathcal{F}$ is approximately equal to a Poisson flow,
$\mathcal{P},$ with a time dependent rate function $\lambda_{\mathcal{N}%
}\left(  T\right)  .$

\item The exit flow from $\mathcal{N}$ -- \textbf{not} Poissonian in general!
-- has a rate function $\gamma_{\mathcal{N}}\left(  T\right)  ,$ which is
smoother than $\lambda_{\mathcal{N}}\left(  T\right)  $ (due to averaging,
taking place at the node $\mathcal{N}$).

\item As a result, the flows $\lambda_{\mathcal{N}}\left(  T\right)  $ at
various nodes $\mathcal{N}$ should converge to constant limits $\bar{\lambda
}_{\mathcal{N}}\approx\frac{1}{T}\int_{0}^{T}\lambda\left(  t\right)  dt$, as
$T\rightarrow\infty,$ the flows to different nodes being almost independent.

\item The above convergence is uniform in the size of the network.

\item Compute the stationary distribution $\hat{\pi}_{\mathcal{N}}$ at
$\mathcal{N},$ corresponding to the inflow $\mathcal{P}(\bar{\lambda
}_{\mathcal{N}}).$ (These computations are the subject of classical queuing
theory and usually provide explicit formulas.) The claim is that $\hat{\pi
}_{\mathcal{N}}\approx\pi_{\mathcal{N}}.$
\end{itemize}

The Poisson Hypothesis is supposed to give a good estimate if the internal
flow to every node $\mathcal{N}$ is a sum of flows from many other nodes, and
each of these flows constitute only a small fraction of the total flow to
$\mathcal{N}.$

Clearly, the Poisson Hypothesis can not be literally true. It can hopefully
hold only after some kind of ``thermodynamic''\ limit is taken. Its meaning is
that in the long run the different nodes become virtually independent, i.e.
propagation of chaos takes place. The reason for that should be that any
synchronization of the nodes, if initially present, dissolves with time, due
to the randomness of the service times.

For some time it was believed that the Poisson Hypothesis behavior is a
general characteristic of all large highly connected networks. It was proven
in some special cases in \cite{St1} and \cite{RSh1}. However, the
counterexamples were also found recently (see \cite{RSh2}, where some special
service times were considered, and especially \cite{RShV1}, where the network
with exponential service times is constructed, for which the PH breaks down at
high load). We think that the situation here resembles the one of statistical
mechanics, where \textit{all} the models behave alike at high temperatures,
while at low temperatures \textit{some} of them exhibit phase transition
behavior. In the network context the role of the temperature is played by the
average load per server, $\rho,$ and the example described in \cite{RShV1} is
the example of such (first-order) phase transition at high load, where the
corresponding infinite system has multiple equilibrium states. In the present
paper we pursue this analogy further, by showing that under very general
conditions the Poisson Hypothesis holds for general mean-field type closed
networks described below in the low load regime. In the forthcoming paper
\cite{RShV2} we will prove similar results for open networks.

\subsection{The elementary network}

Let $G=\left(  V,E\right)  $ be a finite graph. We want to think about the
graph $G$ as a network of servers, serving clients. So we suppose that $G$ is
endowed with some extra structures. First, at every server $v\in V$ there
might be clients of different nature, so we associate to every $v$ a (finite)
set of types -- or \textit{colors} -- $\left\{  c\in\mathcal{C}\left(
v\right)  \right\}  .$ So we introduce the disjoint union $\bar{V}=\cup_{v\in
V}\mathcal{C}\left(  v\right)  $ of all possible types of clients. We connect
a pair $\left(  v_{1},c_{1}\right)  $ to a pair $\left(  v_{2},c_{2}\right)
,$ where $c_{i}\in\mathcal{C}\left(  v_{i}\right)  ,$ by a \textit{directed}
bond, $e\in\bar{E},$ if $\left(  v_{1},v_{2}\right)  $ is an edge in $E,$ and
moreover if it can happen that a client of type $c_{1}$ is served by the
server $v_{1}$ and as a result is sent to the server $v_{2}$ as a type $c_{2}$ client.

Thus far we just described another (bigger, and directed) graph, $\bar
{G}=\left(  \bar{V},\bar{E}\right)  .$ Next, we have to specify the service
times. We suppose them to have exponential distribution, with the rates
$\gamma\left(  v,c\right)  >0.$

The last piece of information we want to have on $G$ is the transition
probability matrix, $P=\left\Vert P\left[  \left(  v_{1},c_{1}\right)
,\left(  v_{2},c_{2}\right)  \right]  \right\Vert .$ The number $P\left[
\left(  v_{1},c_{1}\right)  ,\left(  v_{2},c_{2}\right)  \right]  $ is the
probability that the client of type $c_{1}$ at the node $v_{1}$ will go to the
node $v_{2}$ as a type $c_{2}$ client.

The condition we want to impose on our \textit{elementary network }$\left(
\bar{G},\gamma\right)  $ is that of \textit{connectedness:}

\begin{condition}
\label{csc} Let us consider a continuous time Markov process on $G,$
corresponding to the \textit{Case of a Single Client. }That means that we have
just one client in our network. Initially it is sitting at some server $v,$
having some color $c\in\mathcal{C}\left(  v\right)  .$ As time goes, the
client is changing his location and color randomly, with the rates given by
the rate function $\gamma\left(  \cdot,\cdot\right)  ,$ and transition
probabilities $P.$ We want this (continuous time finite state) Markov process
to be ergodic. (This is equivalent to the ergodicity of the Markov chain
defined by $P.$)
\end{condition}

In the networks that we are going to consider, there will be many clients, so
we need to describe what happens if several clients come for the service to
the same server. Here we can treat a fairly general situation. Suppose that at
the server $v$ at time $t$ we have a queue $x_{v}=\left\{  c_{1}%
,c_{2},...,c_{k}\right\}  $ of clients, $c_{i}\in\mathcal{C}\left(  v\right)
.$ That means that these clients came to $v$ before $t,$ and are still there
at the moment $t,$ waiting to be served. They are listed in the order of their
arrival. The number $k=k\left(  x_{v}\right)  $ is called the queue length (at
$v$ at the moment $t$). The protocol $R\left(  v\right)  $ is a rule
$i_{v}\left(  \ast\right)  $ for the server $v,$ which assigns to every
non-empty queue $x_{v}$ the index $1\leq i_{v}\left(  x_{v}\right)  \leq k,$
which is the index of the client in $x_{v}$ who is served at the moment $t.$
Once this client finishes its service and leaves the server, the queue $x_{v}$
turns into
\begin{equation}
x_{v}\ominus\left\{  i_{v}\left(  x_{v}\right)  ,c_{i_{v}\left(  x_{v}\right)
}\right\}  \equiv\left\{  c_{1},c_{2},...,c_{i_{v}\left(  x_{v}\right)
-1},c_{i_{v}\left(  x_{v}\right)  +1},...,c_{k}\right\}  , \label{20}%
\end{equation}
so $k\left(  x_{v}\ominus\left\{  i_{v}\left(  x_{v}\right)  ,c_{i_{v}\left(
x_{v}\right)  }\right\}  \right)  =k\left(  x_{v}\right)  -1.$

For example, the server can just serve the clients in the order they arrive,
i.e. $i_{v}\left(  x_{v}\right)  \equiv1$. This protocol is called FIFO --
First-In-First-Out. We will treat fairly general protocols, with two restrictions.

The first is that the server can not be idle if there are clients waiting for
the service. This is called \textit{conservative }discipline\textit{. }

The second is the following \textit{monotonicity} property. Let $c_{1}^{t_{1}%
},c_{2}^{t_{2}},...,c_{k}^{t_{k}},...$ be the schedule of arrival of customers
to the node $v,$ with $t_{i}$ being the moment of arrival of the $i$-th
customer. Suppose we know the service time needed for every client. Then the
protocol $R\left(  v\right)  $ allows us to define the function $N\left(
t\right)  ,$ which is the length of the queue at the moment $t\geq0.$ Consider
now another arrival schedule, which differs by exactly one extra client
$\mathfrak{c}$: $c_{1}^{t_{1}},c_{2}^{t_{2}},...,c_{i}^{t_{i}},\mathfrak{c}%
^{\mathfrak{t}},c_{i+1}^{t_{i+1}},...,c_{k}^{t_{k}},...,$ where $t_{i}%
<\mathfrak{t}<t_{i+1}.$ Then the queue length function would change to a
different one, $N^{\mathfrak{c}}\left(  t\right)  .$ We need that
$N^{\mathfrak{c}}\left(  t\right)  \geq N\left(  t\right)  $ for all $t>0.$
This property holds for most of the natural disciplines.

(Note that it is allowed that the service of the client $c$ is interrupted
once a client with higher priority arrives. The service of $c$ is then resumed
according to the priority rule $R\left(  v\right)  $.)

\subsection{Mean-field type graphs}

Now we associate to the graph $G=\left(  V,E\right)  $ the sequence
$\mathcal{G}_{1}\left(  G\right)  =G\subset\mathcal{G}_{2}\left(  G\right)
\subset...\subset\mathcal{G}_{M}\left(  G\right)  \subset...$ of graphs, which
is constructed as follows. For every $M$ consider a disjoint union of $M$
copies $G^{i}=\left(  V^{i},E^{i}\right)  $ of the graph $G,$ $i=1,...,M.$
Then the graph $\mathcal{G}_{M}=\left(  \mathcal{V}_{M},\mathcal{E}%
_{M}\right)  $ has for its vertices the set $\mathcal{V}_{M}=\cup_{i=1}%
^{M}V^{i}$ of all the vertices of these $M$ copies of $G.$ We declare a pair
$\left(  v_{1}^{i},v_{2}^{j}\right)  \subset\mathcal{V}_{M}$ to be a bond in
$\mathcal{E}_{M},$ $v_{1}^{i}\in V^{i},v_{2}^{j}\in V^{j},$ $i,j=1,...,M$ iff
two vertices $v_{1},v_{2}\in V$ are connected by an edge $e\in E.$ Thus, every
bond of $G$ produces exactly $M^{2}$ bonds of $\mathcal{E}_{M}.$ There are no
other bonds in $\mathcal{E}_{M}.$

For example, if $G$ consists of just one vertex $v,$ connected to itself by a
loop, then $\mathcal{G}_{M}\left(  G\right)  $ will be a complete graph with
$M$ vertices (plus to every vertex there is a loop attached).

We now define the mean field graphs $\mathcal{G}_{M}\left(  \bar{G}\right)  $
in precisely the same way as above. Of course, these graphs have natural
orientations, inherited from $\bar{G}.$ We keep the rate functions the same,
and we define the transition probability matrix $P_{M}$ by
\[
P_{M}\left[  \left(  v_{1}^{i},c_{1}\right)  ,\left(  v_{2}^{j},c_{2}\right)
\right]  =\frac{1}{M}P\left[  \left(  v_{1},c_{1}\right)  ,\left(  v_{2}%
,c_{2}\right)  \right]  .
\]

In what follows we will fix the elementary network\textit{ }$\left(  \bar
{G},\gamma\right)  ,$ and we will be interested in the corresponding
mean-field networks $\mathcal{G}_{M}\left(  \bar{G}\right)  ,$ populated by
$N=\lfloor\rho M\rfloor$ clients. We will call the parameter $\rho$ the
\textit{load}. We thus have for every $M$ the ergodic Markov process on the
network $\mathcal{G}_{M}\left(  \bar{G}\right)  $ with $N$ clients; let us
denote this process by $\nabla_{M}^{\rho},$ while $\pi_{M}^{\rho}$ will denote
its invariant measure. We will be interested in the asymptotic properties of
the measures $\pi_{M}^{\rho}$ as $M\rightarrow\infty,$ as well as in the
character of the convergence $\nu_{M}^{\rho}\left(  t\right)  \rightarrow
\pi_{M}^{\rho}$ of the state $\nu_{M}^{\rho}\left(  t\right)  $ of the process
$\nabla_{M}^{\rho}$ at time $t$ to its limit $\pi_{M}^{\rho}.$

Informally speaking, the Poisson Hypothesis for the mean-field networks
$\mathcal{G}_{M}\left(  \bar{G}\right)  $ is the following statement about the
behavior of the measures $\nu_{M}^{\rho}\left(  t\right)  $ for large $M:$ if
the initial state $\nu_{M}^{\rho}\left(  0\right)  $ is chosen reasonably --
which means that the initial queues at every node do not exceed some constant
$K$ -- then after some time $T\left(  K\right)  ,$ \textit{independent of
}$M,$

\begin{itemize}
\item the $M\left\vert V\right\vert $ servers of our network become almost
independent, i.e. the measure $\nu_{M}^{\rho}\left(  t\right)  $ is close to a
product measure over the set $\mathcal{V}_{M};$

\item at every node $v\in\mathcal{V}_{M}$ the process $\nu_{M}^{\rho}\left(
t\right)  $ looks as if the inflows of all possible types of customers to $v$
are independent Poisson flows $\mathcal{P}_{c}$, $c\in\mathcal{C}\left(
v\right)  ,$ with \textit{constant} rates $\lambda_{c}$ (depending only on
$\rho,$ the graph $\bar{G}$ and the function $\gamma,$ \textit{but not on }$M$
\textit{or }$t$), which customers are then queuing at $\nu$ and are served
according to the service rule $R\left(  v\right)  $ at $v,$ and leaving the
node after being served. Let us denote the stationary state of such a node by
$\chi^{\left\{  \lambda_{c},c\in\mathcal{C}\left(  v\right)  \right\}  }.$
\end{itemize}

\noindent In such a case we clearly have the relations
\begin{equation}
\lambda_{c^{\prime}}=\sum_{v\in V}\sum_{c\in\mathcal{C}\left(  v\right)
}\lambda_{c}P\left[  \left(  v,c\right)  ,\left(  v^{\prime},c^{\prime
}\right)  \right]  , \label{01}%
\end{equation}
where $c^{\prime}\in\mathcal{C}\left(  v^{\prime}\right)  .$ They define the
set $\bar{\lambda}=\left\{  \lambda_{c},c\in\bar{V}\right\}  $ of the rates of
the Poisson inflows up to a common factor, $\alpha.$ Note that the expected
number of customers, $N=N\left(  \bar{\lambda}\right)  $, present in the
network $G$ in the state $\prod_{v\in V}\chi^{\left\{  \lambda_{c}%
,c\in\mathcal{C}\left(  v\right)  \right\}  }$, considered as a function of
$\alpha,$ i.e. $N\left(  \alpha\bar{\lambda}\right)  ,$ is continuous strictly
increasing in $\alpha$ once $\bar{\lambda}\not \equiv 0.$ (Here we need the
monotonicity property of the service discipline $R\left(  v\right)  $).
Therefore the rates $\bar{\lambda}^{\rho}$ are uniquely defined by the
relations $\left(  \ref{01}\right)  ,$ supplemented by the equation $N\left(
\bar{\lambda}\right)  =\rho.$ The stationary distribution of the queues,
corresponding to the above Poisson inflows $\bar{\lambda}^{\rho}$ will be
denoted by $\chi^{\rho}:$%
\begin{equation}
\chi^{\rho}=\prod_{v\in V}\chi^{\left\{  \lambda_{c}^{\rho},c\in
\mathcal{C}\left(  v\right)  \right\}  }. \label{05}%
\end{equation}
To save on notation, we will consider below the case when the service
discipline is FIFO with priorities. To describe it we need to endow every set
of colors $\mathcal{C}\left(  v\right)  $ with priority relation
$\succcurlyeq,$ which is a linear order, and we say that the client of type
$c_{1}$ has higher priority than $c_{2}$ iff $c_{2}\succcurlyeq c_{1}.$ At
every moment when there are several clients waiting at a server for the
service, the client with the highest priority is served. If there are several
such clients, they are served in the order they came to the server. The
priorities are respected to such an extent that even when a high priority
client arrives at the time when a low priority client is under the service,
then the service of the latter is interrupted, and is resumed only after the
former one leaves the node. The (notational) advantage of such discipline is
that the queue can be described by just a vector in $\Omega=\mathbb{Z}%
^{\left\vert \mathcal{C}\left(  v\right)  \right\vert }.$ But otherwise our
proof can be carried out literally.

Due to the symmetry of our network we can consider all the measures $\nu
_{M}^{\rho}\left(  t\right)  $ and $\pi_{M}^{\rho}$ to be the measures on the
same space $\mathcal{M}\left(  \Omega\right)  ,$ interpreting them as the
frequency of a given state of a given server in the network $\mathcal{G}%
_{M}\left(  \bar{G}\right)  $. Thus, $\nu_{M}^{\rho}\left(  t\right)  $ and
$\pi_{M}^{\rho}$ are elements of $\mathcal{M}\left(  \mathcal{M}\left(
\Omega\right)  \right)  .$ However, the limiting measure $\nu_{\infty}^{\rho
}\left(  t\right)  ,$ unlike all the other measures above, is again the
element of $\mathcal{M}\left(  \Omega\right)  .$ The process $\nu_{\infty
}^{\rho}\left(  t\right)  $ is a Non-Linear Markov Process (NLMP). See
\cite{RSh1} for more details.

Using these notations, the above claims mean that the convergence $\nu
_{M}^{\rho}\left(  t\right)  \rightarrow\pi_{M}^{\rho}$ is uniform in $M,$ and
that $\pi_{M}^{\rho}\rightarrow\chi^{\rho}$ as $M\rightarrow\infty.$ The
measure $\chi^{\rho}$ is called \textquotedblleft the Poisson Hypothesis
behavior\textquotedblright.

\section{The main result}

Suppose we are given the elementary network $\left(  \bar{G},\gamma\right)  ,$
which is connected in the sense of Condition \ref{csc}. We suppose that every
server $v\in V$ is supplied with the conservative service rule $R\left(
v\right)  $ (which allows the server to decide the order in which the queuing
clients are served). Consider now the sequence $\nabla_{M}^{\rho}$ of Markov
process on the networks $\mathcal{G}_{M}\left(  \bar{G}\right)  ,$
$M=1,2,...,$ populated by $N=\lfloor\rho M\rfloor$ clients. Let $\pi_{M}%
^{\rho}$ stands for its invariant measure.

\begin{theorem}
\label{t1} There exists the value $\rho=\rho\left(  \bar{G},\gamma,P\right)
,$ such that for any $\rho<\rho\left(  \bar{G},\gamma.P\right)  $ the weak
limit $\lim_{M\rightarrow\infty}\pi_{M}^{\rho}$ exists and is equal to the
measure $\chi^{\rho}.$
\end{theorem}

Our next result deals with the convergence to the stationary states $\pi
_{M}^{\rho}.$ For that, we have to specify the class of the initial states of
our Markov processes.

\begin{theorem}
\label{t2} Suppose we start the Markov process $\nabla_{M}^{\rho}$ in the
state $\nu_{M}^{\rho}\left(  0\right)  ,$ which has the property that at every
node $v\in\mathcal{V}_{M}$ the queue length $k_{v}=k\left(  v,t=0\right)  $
satisfies
\begin{equation}
\int\exp\left\{  \varkappa k_{v}\right\}  ~d\nu_{M}^{\rho}\left(  0\right)
\leq K, \label{92}%
\end{equation}
where $\varkappa>0,K$ are some constants, defined below. Suppose that
$\rho<\rho\left(  \bar{G},\gamma\right)  .$ Then for every local function $f$
we have
\[
\left\vert \int f~d\nu_{M}^{\rho}\left(  t\right)  -\int f~d\pi_{M}^{\rho
}\right\vert \leq\left\Vert f\right\Vert \exp\left\{  -\tau t\right\}  ,
\]
where the constant $\tau=\tau\left(  \varkappa,K\right)  $ does not depend on
$M$.
\end{theorem}

\section{Proof}

The plan of the proof is the following. First, we establish the exponential
convergence of the limiting process, $\nu_{\infty}^{\rho}\left(  t\right)  ,$
which is NLMP, to its limiting state. We will see that the limiting state --
$\chi^{\rho}$ -- is unique, and coincides with the Poisson point. Since by the
Khasminsky theorem -- see Theorem 1.2.14 in \cite{L} -- any accumulation point
of the sequence $\pi_{M}^{\rho}$ is a stationary measure of $\nu_{\infty
}^{\rho}\left(  t\right)  ,$ while the family $\pi_{M}^{\rho},$ $M=1,2,...$ is
compact, that will prove all our claims.

In this paper we present a part of the proof of our theorems for a concrete
network studied in \cite{RShV1}. The generalization to other networks is straightforward.

The graph $G$ in \cite{RShV1} has three vertices, $\bar{O},\bar{A}$ and
$\bar{B}.$ The clients at $\bar{O}$ have no priorities, i.e. $\left\vert
\mathcal{C}\left(  \bar{O}\right)  \right\vert =1,$ while the clients at
$\bar{A}$ or $\bar{B}$ are of two types, i.e. $\left\vert \mathcal{C}\left(
\bar{A}\right)  \right\vert =\left\vert \mathcal{C}\left(  \bar{B}\right)
\right\vert =2.$ In notations of \cite{RShV1}, $\mathcal{C}\left(  \bar
{A}\right)  =\left\{  A,BA\right\}  ,$ $\mathcal{C}\left(  \bar{B}\right)
=\left\{  B,AB\right\}  ,$ and $BA$-clients have priority over $A$-clients (at
$\bar{A}$), while $AB$-clients -- over $B$-clients (at $\bar{B}$). The rates
$\gamma\left(  \cdot,\cdot\right)  $ are the following:%
\begin{align}
\gamma\left(  \bar{O}\right)   &  \equiv\gamma_{O}=3,\label{90}\\
\gamma\left(  \bar{A},BA\right)   &  \equiv\gamma_{BA}=\gamma\left(  \bar
{B},AB\right)  \equiv\gamma_{AB}=2,\nonumber\\
\gamma\left(  \bar{A},A\right)   &  \equiv\gamma_{A}=\gamma\left(  \bar
{B},B\right)  \equiv\gamma_{B}=10.\nonumber
\end{align}
The matrix $P$ is given by%
\begin{equation}
P=\left[
\begin{array}
[c]{cccccc}
& O & A & BA & B & AB\\
O & o & \frac{1}{2} & o & \frac{1}{2} & o\\
A & o & o & o & o & 1\\
BA & 1 & o & o & o & o\\
B & o & o & 1 & o & o\\
AB & 1 & o & o & o & o
\end{array}
\right]  . \label{91}%
\end{equation}
It is clearly ergodic.

\subsection{Compactness}

We first prove the following compactness statement.

\begin{lemma}
Suppose all the initial states $\nu_{M}^{\rho}\left(  0\right)  $ of our
processes satisfy the condition $\left(  \ref{92}\right)  ,$ $M=1,2,...$ .
Suppose that the load $\rho$ is small enough. Then there exist two values
$\varkappa^{\prime}$ and $K^{\prime},$ such that for all nodes $v,$ all $t$
and $M$
\begin{equation}
\int\exp\left\{  \varkappa^{\prime}k_{v}\right\}  ~d\nu_{M}^{\rho}\left(
t\right)  \leq K^{\prime}. \label{94}%
\end{equation}

\end{lemma}

\begin{proof}
Note that the rate of the flow of the clients of type $c_{1}$ from the server
$v_{1}$ to the server $v_{2},$ where they turn into the clients of type
$c_{2},$ is bounded from above by $\rho\gamma\left(  v,c\right)  P\left[
\left(  v_{1},c_{1}\right)  ,\left(  v_{2},c_{2}\right)  \right]  .$ That
implies $\left(  \ref{94}\right)  $ for $\rho$ small enough.
\end{proof}

Suppose we know the convergence of the generators of the Markov processes
$\nabla_{M}^{\rho}$ to that of the limiting Non-Linear Markov Process
$\nabla_{\infty}^{\rho}.$ That implies the convergence of the processes
$\nabla_{M}^{\rho}$ to $\nabla_{\infty}^{\rho}$ on any finite time interval.
Suppose also that we can show that the NLMP $\nabla_{\infty}^{\rho}$ is
ergodic -- i.e. for every its initial state $\nu_{\infty}^{\rho}\left(
0\right)  $ satisfying $\left(  \ref{92}\right)  $ we have the convergence to
the (unique) limiting state, $\nu_{\infty}^{\rho}\left(  t\right)
\rightarrow\pi_{\infty}^{\rho},$ which is uniform in the choice of
$\nu_{\infty}^{\rho}\left(  0\right)  .$ Then, first of all, we have the
convergence $\pi_{M}^{\rho}\rightarrow\nu_{\infty}^{\rho}\left(
\infty\right)  $ as $M\rightarrow\infty.$ Indeed, the family $\pi_{M}^{\rho}$
is compact, due to $\left(  \ref{94}\right)  .$ Since by the Khasminsky
theorem -- see Theorem 1.2.14 in \cite{L} -- any accumulation point of the
sequence $\pi_{M}^{\rho}$ is a stationary measure of $\nabla_{\infty}^{\rho},$
while the latter is unique, the convergence follows.

Let us check that the convergence $\nu_{M}^{\rho}\left(  t\right)
\rightarrow\pi_{M}^{\rho}$ is uniform in $M.$ Let $\varepsilon>0,$ and let
$T=T\left(  \varepsilon\right)  $ be the time for which the estimate
$\rho_{KROV}\left(  \nu_{\infty}^{\rho}\left(  T\right)  ,\pi_{\infty}^{\rho
}\right)  <\varepsilon$ holds for any initial condition $\nu_{\infty}^{\rho
}\left(  0\right)  $ satisfying $\left(  \ref{92}\right)  .$ We will show that
for all $M$ large enough and for all $t>T\left(  \varepsilon\right)  $ we
have
\begin{equation}
\rho_{KROV}\left(  \nu_{M}^{\rho}\left(  t\right)  ,\pi_{M}^{\rho}\right)
<3\varepsilon. \label{95}%
\end{equation}
Indeed, the state $\nu_{M}^{\rho}\left(  t-T\right)  $ satisfies $\left(
\ref{94}\right)  ,$ and therefore the evolution $\nabla_{\infty}^{\rho},$
applied to it for time $T$, results in a state $\tilde{\nu}_{M}^{\rho}\left(
t\right)  ,$ satisfying $\rho_{KROV}\left(  \tilde{\nu}_{M}^{\rho}\left(
t\right)  ,\pi_{\infty}^{\rho}\right)  <\varepsilon.$ But for $M$ large enough
$\rho_{KROV}\left(  \tilde{\nu}_{M}^{\rho}\left(  t\right)  ,\nu_{M}^{\rho
}\left(  t\right)  \right)  <\varepsilon,$ $\rho_{KROV}\left(  \pi_{M}^{\rho
},\pi_{\infty}^{\rho}\right)  <\varepsilon,$ so $\left(  \ref{95}\right)  $ follows.

So to prove our theorems it remains to show that for any initial state
$\nu_{\infty}^{\rho}\left(  0\right)  $ of the NLMP satisfying $\left(
\ref{92}\right)  $ we have the convergence $\nu_{\infty}^{\rho}\left(
t\right)  $ to $\chi^{\rho}.$

\subsection{The derived process}

We want to study our particle system with many nodes $M,$ or even with
infinite number of them, in the limit as the number of particles, $N,$ or the
ratio $\rho=\frac{N}{M}$ -- goes to zero. Understood literally, the limiting
object is trivial. We want a non-trivial version of it. The object we
construct is very similar to the derivative of the function at the point where
it vanishes, so we look not on the function, but on the properties of its increments.

We will construct this derived process for the case of $M=\infty,$ i.e. for
the NLMP. It can be defined in the following way. Let first $M$ is finite, and
$N$ has some value. Let $\mu_{M,\rho}$ be a state of our Markov process. For
every configuration $\sigma$ of our process we can define a measure
$v_{\sigma}$ on $\mathcal{Z}=\cup_{v\in V}\left(  \mathbb{Z}_{+}^{\left\vert
\mathcal{C}\left(  v\right)  \right\vert }\setminus\mathbf{0}\right)  $ -- the
disjoint union of lattices \textit{without the origins} -- which is the
distribution of \textquotedblleft the queue of the average
customer\textquotedblright. Namely, for every client $c$, forming $\sigma,$ we
consider a point $x_{c}$ in $\mathcal{Z},$ by first taking the server
$v\left(  c\right)  ,$ corresponding to $c,$ and then marking the point
$x_{c}\in\mathbb{Z}_{+}^{\left\vert \mathcal{C}\left(  v\right)  \right\vert
}\setminus\mathbf{0,}$ describing the queue at this server. Then $v_{\sigma
}=\frac{1}{N}\sum_{c\in\sigma}\delta_{x_{c}}.$ We define the derived process
$\mu_{M,\rho}^{\prime}$ by
\[
\mu_{M,\rho}^{\prime}=\int v_{\sigma}~d\mu_{M,\rho}\left(  \sigma\right)  .
\]
Clearly, we can take the limit of $\mu_{M,\rho}^{\prime}$ as $M\rightarrow
\infty,$ keeping $\rho$ fixed. The limiting process -- the \textit{derivative
-- }$\mu_{\infty,\rho}^{\prime}$ will be also NLMP.

It can be described alternatively as follows. The state of the NLMP on the
graph $G=\left(  V,E\right)  $ is represented by the collection of $\left\vert
V\right\vert $ probability measures $\mu_{v},$ $v\in V$, defined,
correspondingly, on lattices $\mathbb{Z}_{+}^{\left\vert \mathcal{C}\left(
v\right)  \right\vert }$. For every $x\in\mathcal{Z}$ we put $\left\vert
x\right\vert =x_{1}+...+x_{d\left(  x\right)  },$ where $d\left(  x\right)  $
is the dimension of $x.$ The derivative is a single probability measure
$\mu^{\prime}$ on the disjoint union $\mathcal{Z}=\cup_{v\in V}\left(
\mathbb{Z}_{+}^{\left\vert \mathcal{C}\left(  v\right)  \right\vert }%
\setminus\mathbf{0}\right)  $ of lattices. It is defined by its density with
respect to various $\mu_{v}$-s: for every $x\in\mathbb{Z}_{+}^{\left\vert
\mathcal{C}\left(  v\right)  \right\vert }\setminus\mathbf{0}\subset
\mathcal{Z}$ we put%
\[
\mu^{\prime}\left(  x\right)  =\frac{1}{N\left(  \mu\right)  }\left\vert
x\right\vert \mu_{v}\left(  x\right)  ,\text{ where }N\left(  \mu\right)
=\sum_{x\in\mathcal{Z}}\mu\left(  x\right)  \equiv\sum_{v\in V}\sum
_{x\in\mathbb{Z}_{+}^{\left\vert \mathcal{C}\left(  v\right)  \right\vert
}\setminus\mathbf{0}}\left\vert x\right\vert \mu_{v}\left(  x\right)  .
\]
Since all our measures have exponential moments, the above normalization is
possible. The recovering of the probability measures $\mu_{v}$ from
$\mu^{\prime}$ can be done only \textquotedblleft up to a
constant\textquotedblright, as usual. We will take care of it below.

In accordance with our choice of variables we introduce the norm
\begin{equation}
\Vert\mu^{\prime}\Vert=\sum_{x\in\mathcal{Z}}\frac{\mu^{\prime}\left(
x\right)  }{\left\vert x\right\vert }e^{|x|}. \label{04}%
\end{equation}
In case of the signed measure $\mu$ on $\mathcal{Z}$, we define $\left\Vert
\mu\right\Vert $ by the same relation, with $\mu\left(  x\right)  $ replaced
by $\left\vert \mu\left(  x\right)  \right\vert .$ Below we are assuming that
all the measures $\mu$ considered are elements of the space $Y$ of measures on
$\mathcal{Z}$ with finite norm $\left\Vert \mu\right\Vert .$

The process $\mu_{\infty,\rho}^{\prime}$ is again a NLMP. Note, however, that
the limit $\mu_{\infty,0}^{\prime}\equiv\lim_{\rho\rightarrow0}\mu
_{\infty,\rho}^{\prime}$ is a usual (linear) Markov process on $\mathcal{Z}$.
It is defined by the following jump rates. Let $v$ be a server, and
$x\in\mathbb{Z}_{+}^{\left\vert \mathcal{C}\left(  v\right)  \right\vert
}\subset\mathcal{Z},$ $x=\left(  x_{1},...,x_{\left\vert \mathcal{C}\left(
v\right)  \right\vert }\right)  .$ We will suppose that the coordinates are
listed in the priority order, so if $x=\left(  x_{1},...,x_{k-1}%
,x_{k},0,...,0\right)  ,$ with $x_{k}>0,$ then after a single act of service
the queue will be $\bar{x}=\left(  x_{1},...,x_{k-1},x_{k}-1,0,...,0\right)
.$ The rate $c_{x\rightarrow\bar{x}}$ of the jump $x\rightarrow\bar{x}$ is
then%
\begin{equation}
c_{x\rightarrow\bar{x}}=\frac{\left\vert x\right\vert -1}{\left\vert
x\right\vert }\sum_{i,j}\gamma\left(  v,k\right)  P\left[  \left(  v,k\right)
,\left(  v_{i},c_{j}\right)  \right]  , \label{02}%
\end{equation}
where the sum is taken over all possible outcomes of the service at $v$ of the
client of $k$-th priority. Another possible jump for the process $\mu
_{\infty,0}^{\prime}$ from $x$ is to the cite $\mathbf{1}_{ij}\equiv\left(
0,...,0,1_{j},0,...,0\right)  \in\mathbb{Z}_{+}^{\left\vert \mathcal{C}\left(
v_{i}\right)  \right\vert }\subset\mathcal{Z};$ it happens with the rate%
\begin{equation}
c_{x\rightarrow\mathbf{1}_{ij}}=\frac{1}{\left\vert x\right\vert }%
\gamma\left(  v,k\right)  P\left[  \left(  v,k\right)  ,\left(  v_{i}%
,c_{j}\right)  \right]  . \label{03}%
\end{equation}
From the definitions $\left(  \ref{02}\right)  -\left(  \ref{03}\right)  $ and
our Condition \ref{csc} it follows that the process $\mu_{\infty,0}^{\prime}$
is ergodic, and its stationary distribution coincides with that of the Markov
process on $G,$ corresponding to the \textit{Case of a Single Client.}

\subsection{Dynamics}

The NLMP dynamics \textit{in }$^{\prime}$\textit{-coordinates} is given by a
differential equation
\begin{equation}
\dot{\mu}(t)=F(\mu(t)),\ \mu\left(  0\right)  =\mu_{0}. \label{E1}%
\end{equation}
The right-hand side of (\ref{E1}) is the sum of two terms; the first one is
linear in $\mu,$ and the second one is quadratic:
\begin{equation}
F(\mu)=G\mu+\rho H(\mu). \label{E5}%
\end{equation}
To convince the reader, we will write down these equations for the system
studied in \cite{RShV1}.

The description of the NLMP in \cite{RShV1} was done via probability measure
$\nu_{t}=\left\{  \nu_{t}\left(  x\right)  ,~x\in\mathbb{Z}^{5}\right\}  .$
Its evolution was given by the equation%

\begin{align*}
&  \frac{d\nu_{t}\left(  x\right)  }{dt}=-\nu_{t}\left(  x\right)  \left(
\sum_{a=O,A,B,AB,BA}\lambda_{a}\left(  t\right)  \right) \\
&  -\nu_{t}\left(  x\right)  \left(  \sum_{a=O,AB,BA}\gamma_{a}\left(
1-\delta_{x_{a}}\right)  +\gamma_{A}\left(  1-\delta_{x_{A}}\right)
\delta_{x_{BA}}+\gamma_{B}\left(  1-\delta_{x_{B}}\right)  \delta_{x_{AB}%
}\right) \\
&  +\nu_{t}\left(  x-\Delta_{O}\right)  \left(  1-\delta_{x_{O}}\right)
\left(  \lambda_{AB}\left(  t\right)  +\lambda_{BA}\left(  t\right)  \right)
+\sum_{a=A,B}\nu_{t}\left(  x-\Delta_{a}\right)  \left(  1-\delta_{x_{a}%
}\right)  \frac{\lambda_{O}\left(  t\right)  }{2}\\
&  +\nu_{t}\left(  x-\Delta_{AB}\right)  \left(  1-\delta_{x_{AB}}\right)
\lambda_{A}\left(  t\right)  +\nu_{t}\left(  x-\Delta_{BA}\right)  \left(
1-\delta_{x_{BA}}\right)  \lambda_{B}\left(  t\right) \\
&  +\sum_{a=O,AB,BA}\nu_{t}\left(  x+\Delta_{a}\right)  \mathbf{\gamma}%
_{a}+\nu_{t}\left(  x+\Delta_{A}\right)  \mathbf{\gamma}_{A}\delta_{x_{BA}%
}+\nu_{t}\left(  x+\Delta_{B}\right)  \mathbf{\gamma}_{B}\delta_{x_{AB}},
\end{align*}
where%
\begin{align*}
\lambda_{a}\left(  t\right)   &  =\gamma_{a}\sum_{x:x_{a}>0}\nu_{t}\left(
x\right)  ,\text{ for }a=O,AB,BA,~\\
\lambda_{A}\left(  t\right)   &  =\gamma_{A}\sum_{x:x_{A}>0,x_{BA}=0}\nu
_{t}\left(  x\right)  ,\ \ \lambda_{B}\left(  t\right)  =\gamma_{B}%
\sum_{x:x_{B}>0,x_{AB}=0}\nu_{t}\left(  x\right)  ,
\end{align*}
and the 5D vectors $\Delta_{a}$ are the basis vectors of the lattice
$\mathbb{Z}^{5}.$ In fact, $\nu_{t}=\nu_{t}^{A}\times\nu_{t}^{B}\times\nu
_{t}^{O},$ where the probability measures $\nu_{t}^{A},\nu_{t}^{B}$ are
defined on $\mathbb{Z}_{+}^{2},$ while $\nu_{t}^{O}$ -- on $\mathbb{Z}_{+}%
^{1},$ and their evolution is given by the relations%
\begin{align*}
\frac{d\nu_{t}^{A}\left(  x\right)  }{dt}  &  =-\nu_{t}^{A}\left(  x\right)
\left(  \frac{\lambda_{O}\left(  t\right)  }{2}+\lambda_{B}\left(  t\right)
+\gamma_{BA}\left(  1-\delta_{x_{BA}}\right)  +\gamma_{A}\left(
1-\delta_{x_{A}}\right)  \delta_{x_{BA}}\right) \\
&  +\nu_{t}^{A}\left(  x-\Delta_{A}\right)  \left(  1-\delta_{x_{A}}\right)
\frac{\lambda_{O}\left(  t\right)  }{2}+\nu_{t}^{A}\left(  x-\Delta
_{BA}\right)  \left(  1-\delta_{x_{BA}}\right)  \lambda_{B}\left(  t\right) \\
&  +\nu_{t}^{A}\left(  x+\Delta_{A}\right)  \mathbf{\gamma}_{A}\delta_{x_{BA}%
}+\nu_{t}^{A}\left(  x+\Delta_{BA}\right)  \mathbf{\gamma}_{BA},
\end{align*}%
\begin{align*}
\frac{d\nu_{t}^{B}\left(  x\right)  }{dt}  &  =-\nu_{t}^{B}\left(  x\right)
\left(  \frac{\lambda_{O}\left(  t\right)  }{2}+\lambda_{A}\left(  t\right)
+\gamma_{AB}\left(  1-\delta_{x_{AB}}\right)  +\gamma_{B}\left(
1-\delta_{x_{B}}\right)  \delta_{x_{AB}}\right) \\
&  +\nu_{t}^{B}\left(  x-\Delta_{B}\right)  \left(  1-\delta_{x_{B}}\right)
\frac{\lambda_{O}\left(  t\right)  }{2}+\nu_{t}^{B}\left(  x-\Delta
_{AB}\right)  \left(  1-\delta_{x_{AB}}\right)  \lambda_{A}\left(  t\right) \\
&  +\nu_{t}^{B}\left(  x+\Delta_{B}\right)  \mathbf{\gamma}_{B}\delta_{x_{AB}%
}+\nu_{t}^{B}\left(  x+\Delta_{AB}\right)  \mathbf{\gamma}_{AB},
\end{align*}%
\begin{align*}
&  \frac{d\nu_{t}^{O}\left(  x\right)  }{dt}=-\nu_{t}^{O}\left(  x\right)
\left(  \sum_{a=AB,BA}\lambda_{a}\left(  t\right)  +\gamma_{O}\left(
1-\delta_{x_{O}}\right)  \right) \\
&  +\nu_{t}^{O}\left(  x-\Delta_{O}\right)  \left(  1-\delta_{x_{O}}\right)
\left(  \lambda_{AB}\left(  t\right)  +\lambda_{BA}\left(  t\right)  \right)
+\nu_{t}^{O}\left(  x+\Delta_{O}\right)  \mathbf{\gamma}_{O}.
\end{align*}
The new $^{\prime}$-variables are introduced as follows: $\nu_{t}^{\prime
O}\left(  x\right)  =\frac{\left\vert x\right\vert \nu_{t}^{O}\left(
x\right)  }{\rho}$ (with $x\in\mathbb{Z}_{+}^{1}$), $\nu_{t}^{\prime A}\left(
x\right)  =\frac{\left\vert x\right\vert \nu_{t}^{A}\left(  x\right)  }{\rho
},$ $\nu_{t}^{\prime B}\left(  x\right)  =\frac{\left\vert x\right\vert
\nu_{t}^{B}\left(  x\right)  }{\rho}$ (with $x\in\mathbb{Z}_{+}^{2}$),
$\left\vert x\right\vert \geq1,$ where $\rho=\sum_{x\in\mathbb{Z}_{+}^{5}%
}\left\vert x\right\vert \nu_{t}\left(  x\right)  $ is the number of
particles, (which is conserved). Substituting, we have for $x\neq\Delta
_{BA},\Delta_{A}$
\begin{align}
\frac{d\nu_{t}^{\prime A}\left(  x\right)  }{dt}  &  =-\rho\nu_{t}^{\prime
A}\left(  x\right)  \left(  \frac{\gamma_{O}}{2}\sum_{x:x_{O}>0}\frac{\nu
_{t}^{\prime O}\left(  x\right)  }{\left\vert x\right\vert }+\gamma_{B}%
\sum_{x:x_{B}>0,x_{AB}=0}\frac{\nu_{t}^{\prime B}\left(  x\right)
}{\left\vert x\right\vert }\right) \nonumber\\
&  +\frac{\rho\frac{\gamma_{O}}{2}\left\vert x\right\vert }{\left\vert
x-\Delta_{A}\right\vert }\nu_{t}^{\prime A}\left(  x-\Delta_{A}\right)
\left(  1-\delta_{x_{A}}\right)  \sum_{x:x_{O}>0}\frac{\nu_{t}^{\prime
O}\left(  x\right)  }{\left\vert x\right\vert }\label{010}\\
&  +\frac{\rho\gamma_{B}\left\vert x\right\vert }{\left\vert x-\Delta
_{BA}\right\vert }\nu_{t}^{\prime A}\left(  x-\Delta_{BA}\right)  \left(
1-\delta_{x_{BA}}\right)  \sum_{x:x_{B}>0,x_{AB}=0}\frac{\nu_{t}^{\prime
B}\left(  x\right)  }{\left\vert x\right\vert }\nonumber\\
&  +\frac{\left\vert x\right\vert }{\left\vert x+\Delta_{A}\right\vert }%
\nu_{t}^{\prime A}\left(  x+\Delta_{A}\right)  \mathbf{\gamma}_{A}%
\delta_{x_{BA}}+\frac{\left\vert x\right\vert }{\left\vert x+\Delta
_{BA}\right\vert }\nu_{t}^{\prime A}\left(  x+\Delta_{BA}\right)
\mathbf{\gamma}_{BA}\nonumber\\
&  -\nu_{t}^{\prime A}\left(  x\right)  \left(  \gamma_{BA}\left(
1-\delta_{x_{BA}}\right)  +\gamma_{A}\left(  1-\delta_{x_{A}}\right)
\delta_{x_{BA}}\right)  .\nonumber
\end{align}
Also,%
\begin{align}
&  \frac{d\nu_{t}^{\prime A}\left(  \Delta_{A}\right)  }{dt}=-\rho\nu
_{t}^{\prime A}\left(  \Delta_{A}\right)  \left(  \frac{\gamma_{O}}{2}%
\sum_{x:x_{O}>0}\frac{\nu_{t}^{\prime O}\left(  x\right)  }{\left\vert
x\right\vert }+\gamma_{B}\sum_{x:x_{B}>0,x_{AB}=0}\frac{\nu_{t}^{\prime
B}\left(  x\right)  }{\left\vert x\right\vert }\right) \nonumber\\
&  -\rho\frac{\gamma_{O}}{2}\sum_{\left\vert x\right\vert >0}\frac{\nu
_{t}^{\prime A}\left(  x\right)  }{\left\vert x\right\vert }\sum_{x:x_{O}%
>0}\frac{\nu_{t}^{\prime O}\left(  x\right)  }{\left\vert x\right\vert
}\label{11}\\
&  +\frac{\gamma_{O}}{2}\sum_{x:x_{O}>0}\frac{\nu_{t}^{\prime O}\left(
x\right)  }{\left\vert x\right\vert }+\frac{1}{2}\nu_{t}^{\prime A}\left(
2\Delta_{A}\right)  \mathbf{\gamma}_{A}+\frac{1}{2}\nu_{t}^{\prime A}\left(
\Delta_{A}+\Delta_{BA}\right)  \mathbf{\gamma}_{BA}-\nu_{t}^{\prime A}\left(
\Delta_{A}\right)  \gamma_{A},\nonumber
\end{align}%
\begin{align}
\frac{d\nu_{t}^{\prime A}\left(  \Delta_{BA}\right)  }{dt}  &  =-\rho\nu
_{t}^{\prime A}\left(  \Delta_{BA}\right)  \left(  \frac{\gamma_{O}}{2}%
\sum_{x:x_{O}>0}\frac{\nu_{t}^{\prime O}\left(  x\right)  }{\left\vert
x\right\vert }+\gamma_{B}\sum_{x:x_{B}>0,x_{AB}=0}\frac{\nu_{t}^{\prime
B}\left(  x\right)  }{\left\vert x\right\vert }\right) \nonumber\\
&  -\rho\gamma_{B}\sum_{\left\vert x\right\vert >0}\frac{\nu_{t}^{\prime
A}\left(  x\right)  }{\left\vert x\right\vert }\sum_{x:x_{B}>0,x_{AB}=0}%
\frac{\nu_{t}^{\prime B}\left(  x\right)  }{\left\vert x\right\vert
}\label{12}\\
&  +\gamma_{B}\sum_{x:x_{B}>0,x_{AB}=0}\frac{\nu_{t}^{\prime B}\left(
x\right)  }{\left\vert x\right\vert }+\frac{1}{2}\nu_{t}^{\prime A}\left(
2\Delta_{BA}\right)  \mathbf{\gamma}_{BA}-\nu_{t}^{\prime A}\left(
\Delta_{BA}\right)  \gamma_{BA}\nonumber
\end{align}
The identical relations hold for the measure $\nu_{t}^{\prime B}.$ For the
measure $\nu_{t}^{\prime O}$ we get
\begin{align}
\frac{d\nu_{t}^{\prime O}\left(  x\right)  }{dt}  &  =\rho\left(
\frac{\left\vert x\right\vert \nu_{t}^{\prime O}\left(  x-1\right)
}{\left\vert x-1\right\vert }-\nu_{t}^{\prime O}\left(  x\right)  \right)
\left[  \gamma_{AB}\sum_{x:x_{AB}>0}\frac{\nu_{t}^{\prime B}\left(  x\right)
}{\left\vert x\right\vert }+\gamma_{BA}\sum_{x:x_{BA}>0}\frac{\nu_{t}^{\prime
A}\left(  x\right)  }{\left\vert x\right\vert }\right] \nonumber\\
&  +\left[  -\nu_{t}^{\prime O}\left(  x\right)  +\frac{\left\vert
x\right\vert }{\left\vert x+1\right\vert }\nu_{t}^{\prime O}\left(
x+1\right)  \right]  \mathbf{\gamma}_{O}\text{ for }\left\vert x\right\vert
>1, \label{13}%
\end{align}%
\begin{align}
\frac{d\nu_{t}^{\prime O}\left(  1\right)  }{dt}  &  =-\rho\left(  \sum
_{x\geq1}\frac{\nu_{t}^{\prime O}\left(  x\right)  }{\left\vert x\right\vert
}+\nu_{t}^{\prime O}\left(  1\right)  \right)  \left[  \gamma_{AB}%
\sum_{x:x_{AB}>0}\frac{\nu_{t}^{\prime B}\left(  x\right)  }{\left\vert
x\right\vert }+\gamma_{BA}\sum_{x:x_{BA}>0}\frac{\nu_{t}^{\prime A}\left(
x\right)  }{\left\vert x\right\vert }\right] \nonumber\\
&  +\left[  -\nu_{t}^{\prime O}\left(  1\right)  +\frac{1}{2}\nu_{t}^{\prime
O}\left(  2\right)  \right]  \mathbf{\gamma}_{O}+\left[  \gamma_{AB}%
\sum_{x:x_{AB}>0}\frac{\nu_{t}^{\prime B}\left(  x\right)  }{\left\vert
x\right\vert }+\gamma_{BA}\sum_{x:x_{BA}>0}\frac{\nu_{t}^{\prime A}\left(
x\right)  }{\left\vert x\right\vert }\right]  . \label{14}%
\end{align}
As we see, the quadratic terms all have the factor $\rho$ in front of them.
Putting it to zero, we get the linear part of the evolution:%
\begin{align}
\frac{d\nu_{t}^{\prime A}\left(  x\right)  }{dt}  &  =\frac{\left\vert
x\right\vert }{\left\vert x\right\vert +1}\nu_{t}^{\prime A}\left(
x+\Delta_{A}\right)  \mathbf{\gamma}_{A}\delta_{x_{BA}}+\frac{\left\vert
x\right\vert }{\left\vert x\right\vert +1}\nu_{t}^{\prime A}\left(
x+\Delta_{BA}\right)  \mathbf{\gamma}_{BA}\label{15}\\
&  -\nu_{t}^{\prime A}\left(  x\right)  \left(  \gamma_{BA}\left(
1-\delta_{x_{BA}}\right)  +\gamma_{A}\left(  1-\delta_{x_{A}}\right)
\delta_{x_{BA}}\right) \nonumber
\end{align}
for $x\neq\Delta_{BA},\Delta_{A}.$ Also,%

\begin{align}
\frac{d\nu_{t}^{\prime A}\left(  \Delta_{A}\right)  }{dt}  &  =\frac
{\gamma_{O}}{2}\sum_{x:x_{O}>0}\frac{\nu_{t}^{\prime O}\left(  x\right)
}{\left\vert x\right\vert }+\frac{1}{2}\nu_{t}^{\prime A}\left(  2\Delta
_{A}\right)  \mathbf{\gamma}_{A}\label{16}\\
&  +\frac{1}{2}\nu_{t}^{\prime A}\left(  \Delta_{A}+\Delta_{BA}\right)
\mathbf{\gamma}_{BA}-\nu_{t}^{\prime A}\left(  \Delta_{A}\right)  \gamma
_{A},\nonumber
\end{align}%
\begin{equation}
\frac{d\nu_{t}^{\prime A}\left(  \Delta_{BA}\right)  }{dt}=\gamma_{B}%
\sum_{x:x_{B}>0,x_{AB}=0}\frac{\nu_{t}^{\prime B}\left(  x\right)
}{\left\vert x\right\vert }+\frac{1}{2}\nu_{t}^{\prime A}\left(  2\Delta
_{BA}\right)  \mathbf{\gamma}_{BA}-\nu_{t}^{\prime A}\left(  \Delta
_{BA}\right)  \gamma_{BA}. \label{17}%
\end{equation}
Exchanging the indices $A\longleftrightarrow B,$ we get the equations for
$\nu_{t}^{\prime B}.$ As for $\nu_{t}^{\prime O},$ we have%

\begin{equation}
\frac{d\nu_{t}^{\prime O}\left(  x\right)  }{dt}=\left[  -\nu_{t}^{\prime
O}\left(  x\right)  +\frac{\left\vert x\right\vert }{\left\vert x+1\right\vert
}\nu_{t}^{\prime O}\left(  x+1\right)  \right]  \mathbf{\gamma}_{O}\text{ for
}\left\vert x\right\vert >1, \label{18}%
\end{equation}%
\begin{equation}
\frac{d\nu_{t}^{\prime O}\left(  1\right)  }{dt}=\left[  -\nu_{t}^{\prime
O}\left(  1\right)  +\frac{1}{2}\nu_{t}^{\prime O}\left(  2\right)  \right]
\mathbf{\gamma}_{O}+\sum_{a=AB,BA}\left[  \gamma_{a}\sum_{x:x_{a}>0}\frac
{\nu_{t}^{\prime a}\left(  x\right)  }{\left\vert x\right\vert }\right]
\text{ for }x=1. \label{19}%
\end{equation}

We want to establish the contraction property of the evolution $\left(
\ref{E5}\right)  $ -- so, in particular, the evolution $\left(  \ref{010}%
\right)  -\left(  \ref{14}\right)  .$

To exhibit it we will compare it with the linear evolution
\begin{equation}
\dot{\mu}^{l}(t)=G\mu^{l}(t),\ \mu^{l}\left(  0\right)  =\mu_{0}, \label{E4}%
\end{equation}
with the same initial point. This is the Markov process with rates $\left(
\ref{02}\right)  -\left(  \ref{03}\right)  ,$ discussed above. In our example
the linear evolution is given by $\left(  \ref{15}\right)  -\left(
\ref{19}\right)  .$ Its informal description is the following. Consider the
infinite network $\mathcal{G}_{\infty},$ which is populated by infinitely many
clients, having nevertheless zero density: $\rho=0$. The initial distribution
$\mu_{0}$ of the queue seen by an average client can be arbitrary. The
dynamics is the following: after waiting in the queue and then being served,
the client goes to a next server of corresponding type -- but he finds it free
with probability one! The reason is the vanishing density: it is improbable
that a client will get to an occupied server. Therefore in the limit
$t\rightarrow\infty$ the measure $\mu^{l}(t)$ is concentrated only on points
$x\in\mathcal{Z}$ with property $\left\vert x\right\vert =1.$

This linear evolution is contracting in the following sense: for every $K>0$
there exists the time $T=T\left(  K\right)  <\infty,$ such that for every two
trajectories $\mu_{1}^{l}(t)$ and $\mu_{2}^{l}(t)$ with $\left\Vert \mu
_{i}^{l}(0)\right\Vert \leq K$ we have for all $t\geq0$
\begin{equation}
\Vert\mu_{1}^{l}(t+T)-\mu_{2}^{l}(t+T)\Vert\leq\frac{1}{2}\Vert\mu_{1}%
^{l}(t)-\mu_{2}^{l}(t)\Vert. \label{E6}%
\end{equation}

Let us define on the subspace $Y_{0}=\{\nu\in Y:\nu\left(  \mathcal{Z}\right)
=0\}$ a new norm:
\begin{equation}
\left\Vert \nu\right\Vert _{1}=\int_{0}^{\infty}e^{\beta t}\left\Vert \nu
^{l}(t)\right\Vert dt, \label{E7}%
\end{equation}
where $\nu^{l}(t)$ is the solution of (\ref{E4}) with $\nu^{l}(0)=\nu$. If
$\beta>0$ is small enough, the norm $\left\Vert \cdot\right\Vert _{1}$ is
finite on $Y_{0}$ and equivalent there to $\left\Vert \cdot\right\Vert $.
Moreover (unlike the norm $\left\Vert \cdot\right\Vert $ !), it satisfies the
infinitesimal version of $\left(  \ref{E6}\right)  :$ for any $t\geq0$
\[
\frac{d}{dt}\Vert\mu_{1}^{l}(t)-\mu_{2}^{l}(t)\Vert_{1}\leq-\beta\Vert\mu
_{1}^{l}(t)-\mu_{2}^{l}(t)\Vert_{1}.
\]

\subsection{Non-linear part}

The quadratic part of (\ref{E1}) can be written as
\[
\left[  H(\mu)\right]  \left(  z\right)  =\sum_{x,y\in\mathcal{Z}}%
\mathbf{v}_{xy}\left(  z\right)  \frac{\mu\left(  x\right)  }{\left\vert
x\right\vert }\frac{\mu\left(  y\right)  }{\left\vert y\right\vert },
\]
where $z\in\mathcal{Z}$, and for every pair $x,y\in\mathcal{Z}$ the function
$\mathbf{v}_{xy}\left(  z\right)  $ on $\mathcal{Z}\ $has finite support, the
size of which depends only on the structure of the network $\bar{G}$.

To see this, let us fix a pair of points $x,y\in\mathcal{Z},$ and mark the
equations $\frac{d\nu_{t}^{\prime\ast}\left(  z\right)  }{dt}=...$ from those
listed above, $\left(  \ref{010}\right)  -\left(  \ref{14}\right)  ,$
containing the cross-term $\nu_{t}^{\prime\ast}\left(  x\right)  \nu
_{t}^{\prime\ast}\left(  y\right)  .$ Let us write down the vector
$\mathbf{v}_{xy}=\left\{  \mathbf{v}_{xy}\left(  z\right)  ,z\in
\mathcal{Z}\right\}  .$ This vector is non-zero only if the two indices $x,y$
belong to two different lattices among the three present: $O,A$ or $B.$ For
example, if $x\in O,$ $x\neq1$ and $y\in A,$ $y\neq\Delta_{BA},\Delta_{A},$
then the following six coordinates of $\mathbf{v}_{xy}$ are non-zero:
\[
\mathbf{v}_{xy}\left(  y\right)  =-\frac{\gamma_{O}}{2\left\vert x\right\vert
},\ \mathbf{v}_{xy}\left(  y+\Delta_{A}\right)  =\frac{\gamma_{O}\left\vert
y+\Delta_{A}\right\vert }{2\left\vert x\right\vert \left\vert y\right\vert
},\ \mathbf{v}_{xy}\left(  \Delta_{A}\right)  =-\frac{\gamma_{O}}{2\left\vert
y\right\vert \left\vert x\right\vert },
\]%
\[
\mathbf{v}_{xy}\left(  x\right)  =-\frac{\gamma_{BA}}{\left\vert y\right\vert
},\ \mathbf{v}_{xy}\left(  x+1\right)  =\frac{\left\vert x+1\right\vert
}{\left\vert x\right\vert }\frac{\gamma_{BA}}{\left\vert y\right\vert
},\ \mathbf{v}_{xy}\left(  1\right)  =-\frac{\gamma_{BA}}{\left\vert
x\right\vert \left\vert y\right\vert }%
\]
$\left(  \text{the last three relations require }y_{BA}>0\right)  .$ The rest
of them vanish.

It is also easy to see that for some constant $C>0$ and for each pair
$x,y\in\mathcal{Z}$ we have $\Vert\mathbf{v}_{xy}\left(  \cdot\right)
\Vert\leq Ce^{|x|+|y|}$ (see $\left(  \ref{04}\right)  $).

\subsection{Convergence}

Let us estimate $H(\mu)-H(\nu)$. We have
\[
\Vert H(\mu)-H(\nu)\Vert\leq\sum_{x,y\in\mathcal{Z}}\Vert\mathbf{v}_{xy}%
\Vert\left\vert \frac{\mu_{x}}{\left\vert x\right\vert }\frac{\mu_{y}%
}{\left\vert y\right\vert }-\frac{\nu_{x}}{\left\vert x\right\vert }\frac
{\nu_{y}}{\left\vert y\right\vert }\right\vert \leq
\]%
\[
\leq C\sum_{x,y\in X}e^{|x|+|y|}\left(  \frac{\mu_{x}}{\left\vert x\right\vert
}\left\vert \frac{\mu_{y}}{\left\vert y\right\vert }-\frac{\nu_{y}}{\left\vert
y\right\vert }\right\vert +\frac{\nu_{y}}{\left\vert y\right\vert }\left\vert
\frac{\mu_{x}}{\left\vert x\right\vert }-\frac{\nu_{x}}{\left\vert
x\right\vert }\right\vert \right)  =
\]%
\[
=C\sum_{x,y\in X}\left[  \left(  e^{|x|}\frac{\mu_{x}}{\left\vert x\right\vert
}\right)  \left(  e^{|y|}\frac{\left\vert \mu_{y}-\nu_{y}\right\vert
}{\left\vert y\right\vert }\right)  +\left(  e^{|y|}\frac{\nu_{y}}{\left\vert
y\right\vert }\right)  \left(  e^{|x|}\frac{|\mu_{x}-\nu_{x}|}{\left\vert
x\right\vert }\right)  \right]  \leq
\]%
\begin{equation}
\leq C(\Vert\mu\Vert+\Vert\nu\Vert)\Vert\mu-\nu\Vert. \label{E2}%
\end{equation}

By the equivalence of $\Vert\cdot\Vert$ and $\Vert\cdot\Vert_{1}$, we get the
bound
\begin{equation}
\Vert H(\mu)-H(\nu)\Vert_{1}\leq C_{1}(\Vert\mu\Vert+\Vert\nu\Vert)\Vert
\mu-\nu\Vert_{1}. \label{E8}%
\end{equation}

For our NLMP we have
\begin{equation}
\frac{d}{dt}\Vert\mu(t)-\nu(t)\Vert_{1}\leq-\beta\Vert\mu(t)-\nu(t)\Vert
_{1}+\rho\Vert H(\mu(t))-H(\nu(t))\Vert_{1}. \label{E3}%
\end{equation}
It remains to notice that for every $K>0$ there exists a constant $C_{2}>0$
such for any $\mu(0)\in Y$ with $\left\Vert \mu(0)\right\Vert <K\ $we have
$\Vert\mu(t)\Vert\leq C_{2}K$ for all $t,$ so according to $\left(
\ref{E8}\right)  ,$ the first term beats the second one, once $\rho$ is small.

Now let us go back to our initial NLMP with low load $\rho$. Let
$\varkappa\left(  t\right)  $ be some trajectory of it, with $\left\Vert
\varkappa\left(  t\right)  \right\Vert \leq C_{2}K$, while $\chi^{\rho}$ be
stationary \textquotedblleft Poisson Hypothesis\textquotedblright\ trajectory.
Then the derivative process $\varkappa^{\prime}\left(  t\right)  $ satisfies
$\left(  \ref{E5}\right)  ,$ so by $\left(  \ref{E3}\right)  $ we have
\[
\frac{d}{dt}\Vert\varkappa^{\prime}\left(  t\right)  -\left(  \chi^{\rho
}\right)  ^{\prime}\Vert_{1}\leq\left(  -\beta+\rho C_{1}\left(  \tilde{C}%
_{2}K+\left\Vert \left(  \chi^{\rho}\right)  ^{\prime}\right\Vert \right)
\right)  \Vert\varkappa^{\prime}\left(  t\right)  -\left(  \chi^{\rho}\right)
^{\prime}\Vert_{1}.
\]
Since $\varkappa\left(  t\right)  $ is the only measure \textit{with the load}
$\rho,$ having derivative $\varkappa^{\prime}\left(  t\right)  $, our claim is proven.

\end{document}